\documentclass[letterpaper, 10 pt, conference]{ieeeconf}
\IEEEoverridecommandlockouts
\usepackage[portuguese]{babel}
\usepackage{amsmath,amssymb,amsfonts}
\usepackage{graphicx}
\usepackage{float}
\usepackage{siunitx}
\usepackage{todonotes}
\usepackage{cite}
\usepackage{lipsum}
\usepackage{subcaption}
\usepackage[firstpageonly=true]{draftwatermark} % water mark for preprint

\newtheorem{theorem}{Theorem}

\newtheorem{problem}{Problem}
\newtheorem{assumption}{Assumption}
\newtheorem{remark}{Remark}
\newcommand{\vect}[1]{\mathbf{#1}}
\newcommand{\vectg}[1]{{\boldsymbol{#1}}}
\newcommand\norm[1]{\left\|#1\right\|}

\title{\LARGE \bf
Vision-based Multirotor Control for Spherical Target Tracking:\\
a Bearing-Angle Approach
}

\author{Marcelo~Jacinto
        and Rita~Cunha%,~\IEEEmembership{Member,~IEEE,}% <-this % stops a space
\thanks{The work of M. Jacinto was supported by the PhD Grant 2022.09587.BD from Funda\c{c}{\~a}o para a Ci{\^e}ncia e a Tecnologia
(FCT), Portugal. This work was also supported by FCT, Portugal through LARSyS [DOI: 10.54499/LA/P/0083/2020, 10.54499/UIDP/50009/2020, and  10.54499/UIDB/50009/2020].}% <-this % stops a space
\thanks{The authors are with the Institute for Systems and Robotics (ISR), Laboratory of Robotics and Engineering Systems (LARSyS), Instituto Superior Técnico, University of Lisbon, Portugal. E-mails: {\tt \small \{mjacinto, rita\}@isr.tecnico.ulisboa.pt}.}% <-this % stops a space
}

\begin{document}

\SetWatermarkText{This paper has been accepted for presentation at the \\ 2025 IEEE European Control Conference (ECC)}
\SetWatermarkColor[gray]{0.3}
\SetWatermarkFontSize{0.5cm}
\SetWatermarkAngle{0}
%\SetWatermarkHorCenter{4cm}
\SetWatermarkVerCenter{1.5cm}

% make the title area
\maketitle
\thispagestyle{empty}

% As a general rule, do not put math, special symbols or citations
% in the abstract or keywords.
\begin{abstract}
This work addresses the problem of designing a visual servo controller for a multirotor vehicle, with the end goal of tracking a moving spherical target with unknown radius. To address this problem, we first transform two bearing measurements provided by a camera sensor into a bearing-angle pair. We then use this information to derive the system's dynamics in a new set of coordinates, where the angle measurement is used to quantify a relative distance to the target. Building on this
system representation, we design an adaptive nonlinear control algorithm that takes advantage of the properties of the new system geometry and assumes that the target follows a constant acceleration model. Simulation results illustrate the performance of the proposed control algorithm.
\end{abstract}
\section{INTRODUCTION}
\label{sec:introduction}
The problem of relative target position estimation and tracking has been an active area of research in recent years, with significant progress made in developing robust methods for solving this challenge. This
field of research has numerous practical applications in ground, marine and aerial robotics, such as search-and-rescue operations, surveillance missions, homing on moving docks \cite{Moving_dock_homing}, and more recently aerial cinematography \cite{aerial_cinematography_Bahareh}. These applications typically make use of sensors such as monocular cameras, which provide measurements that can be modeled as bearing vectors \cite{ducted_fan_le_bras}.

Recent works have attempted to address the problem of target tracking by decoupling it into separate estimation, control, and trajectory planning problems \cite{helical_guidance,bearing_angle_approach}. A common approach to solve the estimation problem is to cast the target-tracker nonlinear system into a linear time-varying (LTV) system composed of a single-integrator agent and target moving with constant unknown velocity, in which the agent is only able to measure its direction to the target, but not the range \cite{BATISTA20131065_marine}. Another alternative is to follow a nonlinear adaptive observer design methodology for estimating the position and velocity of an agent relative to an unknown set of landmarks of a target \cite{pe_lebras_2027_springer}. With both methodologies, for such systems to be rendered observable, they are required that the relative trajectories of the target-tracker or tracker-landmarks are persistently exciting (PE), i.e., they yield sufficient variation in the measured bearing vectors over a period of time \cite{BATISTA2011101_generic_observability,BATISTA20131065_marine,pe_lebras_2027_springer}. This requirement constraints significantly the design of trajectories that can be executed by a tracker vehicle \cite{helical_guidance,observalibility_enhancements,anjaly_FIM_TARGET}.

Another approach to solve the problem is to follow an image based visual servo (IBVS) control framework, where the extracted image features are used directly in a sensor based control strategy \cite{ducted_fan_le_bras,Moving_dock_homing,le_bras_velocity_estimation}. Such methods have been used extensively for stabilizing vehicles with respect to planar surfaces and landing on moving planar targets \cite{Moving_dock_homing}. These methods typically do not rely on PE conditions to converge as they do not provide an estimate of the absolute distance to the target. On the other hand, they may rely on scaled relative distance or velocity measurements obtained via optical flow in their control feedback \cite{ducted_fan_le_bras,Moving_dock_homing}.

In many aerial cinematography applications, such as head tracking, the main goal is to capture images of a target from a pre-defined perspective without relying explicitly in direct measurements or estimates of the absolute distance to a target. Bounding boxes are commonly used to enclose detected targets in an image, from which a central bearing pointing towards the target can be extracted. Borrowing inspiration from \cite{bearing_angle_approach}, we can consider that a target can be approximated by a sphere and use the bounding-box surrounding a target to provide extra information in the form of a bearing-angle measurement pair. This additional angle measurement will be inversely proportional to the real distance to the target, and dictate its size in the image plane.

The main contribution of this work is the proposal of an IBVS-inspired nonlinear adaptive control law for tracking a spherical-shaped target with unknown radius from a fixed relative position. We relax the constant velocity model typically considered in the literature, by assuming that the target moves with constant unknown acceleration. To achieve this, we borrow inspiration from \cite{bearing_angle_approach,ducted_fan_le_bras} and transform two bearing measurements into a bearing-angle pair. This new angle measurement plays the crucial role of providing a relative information of the distance between the vehicle and the target, which can be used in practical scenarios, as a reference to dictate the desired size of the target in the image. To derive the system dynamics, we start by taking the time-derivatives of the bearing-angle pair. From this new nonlinear representation the dynamics of the system represented in polar coordinates naturally arise. We then proceed to take advantage of this new coordinate system to derive a control law where the relative size of the target on the image and the relative direction between the vehicle and the target depend on orthogonal components of the relative velocity. Finally, we account for the limited field-of-view of a camera sensor by constraining the input generated by the control law.

The paper is organized as follows: Section \ref{sec:preliminaries} provides mathematical definitions; Section \ref{sec:problem_statement} formulates the target tracking problem; Section \ref{sec:Control_design} focus on control design and stability analysis. Section \ref{sec:target_visibility_constraints} addresses the problem of enforcing visibility constraints, followed by numerical simulations in Section \ref{sec:simulation_results}. Concluding remarks are made in Section \ref{sec:conclusion}.
\section{PRELIMINARIES}
\label{sec:preliminaries}
The identity matrix of $\mathbb{R}^{n \times n}$ is denoted $\vect{I}_{n}$ and $\vect{0}$ denotes a vector of zeros (without underscript). The symbol $\norm{\cdot}$ denotes the euclidean norm. Vectors are lower case and bold while matrices are uppercase. The map  $\vectg{S}(\cdot) : \mathbb{R}^3 \rightarrow \mathfrak{so}(3)$  yields a skew-symmetric matrix $\vectg{S}(\vect{x}) \vect{y} = \vect{x} \times \vect{y} \text{ , } \forall \, \vect{x} ,\vect{y} \in \mathbb{R}^3$. Let $\vectg{S}(\cdot)^{-1}$ be its inverse map. The orthogonal projection operator $\vectg{\Pi}_\vect{y}$ which projects an arbitrary vector $\vect{x} \in \mathbb{R}^{3}$ onto the plane orthogonal to $\vect{y}\in \mathbb{S}^2$ is given by $\vectg{\Pi}_\vect{y}:=\vect{I}_{3} - \vect{y}\vect{y}^\top$.
\section{PROBLEM STATEMENT}
\label{sec:problem_statement}
This section starts by introducing the motion models of the multirotor and the target. It also presents an assumption regarding the target shape and the image features considered, followed by the introduction of the kinematics and dynamics of the features represented as a bearing-angle pair. The formal problem definition then follows.
\subsection{System Modeling}
\label{sec:system_modeling}
Let $\vect{p}_{\mathrm{B}},\vect{v}_{\mathrm{B}} \in \mathbb{R}^3$ and $\mathrm{\vect{R}} \in SO(3)$ denote the position, velocity and orientation of the vehicle body frame $\{\mathcal{B}\}$ with respect to an inertial frame $\{\mathcal{I}\}$, expressed in $\{\mathcal{I}\}$, respectively. The dynamics of the multirotor are given by
\begin{equation}
\begin{split}
    \dot{\vect{p}}_{\mathrm{B}} &= \vect{v}_{\mathrm{B}} \\
    \dot{\vect{v}}_{\mathrm{B}} &= g\vect{e}_{3} - \frac{T}{m} \vect{R}\vect{e}_{3} \\
    \dot{\vect{R}} &= \vect{R} \vectg{S}(\vectg{\omega}),
\end{split}
\end{equation}
where $\vect{e}_{3}=[0,0,1]^\top$ is a unit vector, $m\in \mathbb{R}^{+}$ denotes the total mass of the vehicle, $g \approx 9.8\text{ms}^{-2}$ denotes the gravity acceleration, $T \in \mathbb{R}_{0}^{+}$ the total thrust and $\omega \in \mathbb{R}^3$ the angular-velocity of $\{\mathcal{B}\}$ with respect to $\{\mathcal{I}\}$ expressed in $\{\mathcal{B}\}$.

Analogously, consider the target dynamics to follow a double integrator model with constant acceleration, according to
\begin{equation}
	\begin{cases}
		\dot{\vect{p}}_{\mathrm{T}} &= \vect{v}_{\mathrm{T}} \\
		\dot{\vect{v}}_{\mathrm{T}} &= \vect{a}_{\mathrm{T}} \\
		\dot{\vect{a}}_{\mathrm{T}} &= \vect{0}
	\end{cases},
\end{equation}
where $\vect{p}_{\mathrm{T}},\vect{v}_{\mathrm{T}} \text{ and } \vect{a}_{\mathrm{T}} \in \mathbb{R}^3$ denote position, velocity and acceleration of the target frame $\{\mathcal{T}\}$ with respect to an inertial frame $\{\mathcal{I}\}$, expressed in $\{\mathcal{I}\}$, respectively.

\subsection{Target Image Features}

Consider that the vehicle is equipped with a monocular camera, aligned with $\{\mathcal{B}\}$ such that it can measure bearing outputs ${}^{B} \vect{b}_i  \in  \mathbb{S}^2$ that correspond to the projection of points $\vect{p}_i  \in \mathbb{R}^3$ in the image plane into a virtual unit sphere, expressed in $\{\mathcal{B}\}$ according to
\begin{equation}
    {}^{B} \vect{b}_i := \frac{\vect{p}_i^{img}}{\norm{\vect{p}_i^{img}}}=\frac{\vect{R}^\top(\vect{p}_i -  \vect{p}_{\mathrm{B}})}{\norm{\vect{R}^\top(\vect{p}_i -  \vect{p}_{\mathrm{B}})}} \in  \mathbb{S}^2,
\end{equation}
with $\!\vect{p}_i^{img}\!\!:=\![X_i, Y_i,\!1]^\top \!\! \in \! \mathbb{P}^2\!$ the perspective projection of $\vect{p}_i$.

\begin{assumption}
    The target is approximated by a 3D sphere with unknown constant radius $r \in \mathbb{R}^{+}$. The vehicle measures the bearing vector ${}^{B}\vect{b} \in \mathbb{S}^2$, pointing towards the center of the target and an additional bearing ${}^{B}\vect{b}_t \in \mathbb{S}^2$ tangent to the surface of the sphere, according to Fig. \ref{fig:measurement_model}.
    \label{assumption:spherical_target}
\end{assumption}
\vspace{-0.3cm}
\begin{figure}[H]
	\centering
    \includegraphics[width=0.49\textwidth]{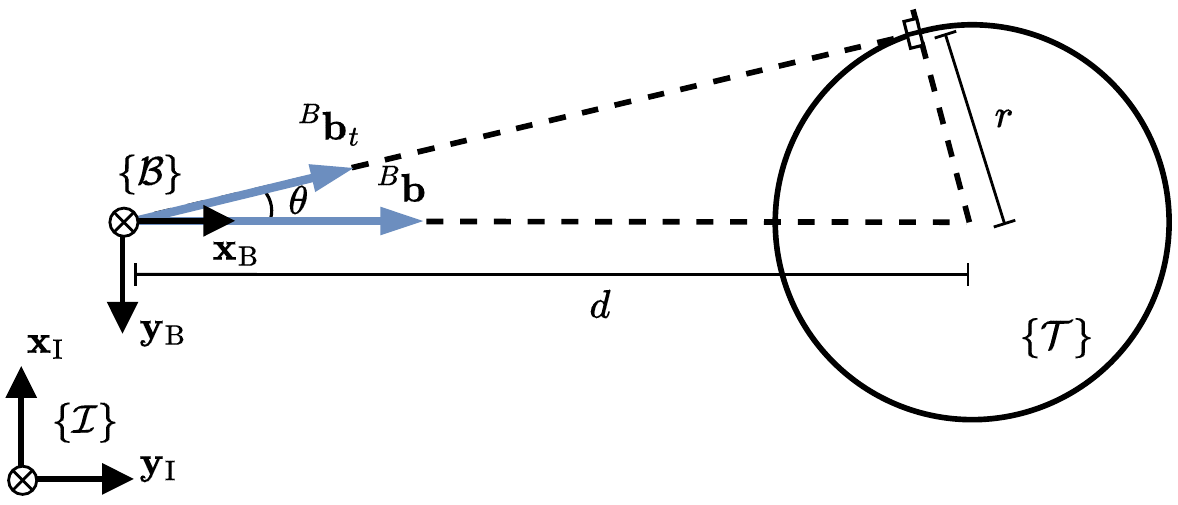}
    \vspace{-0.5cm}
	\caption{Top-down view of a 3-dimensional sphere target observed by the tracker vehicle.}
	\label{fig:measurement_model}
\end{figure}
\vspace{-0.5cm}
Let $\vect{p} := \vect{p}_{\mathrm{T}} - \vect{p}_{\mathrm{B}}$ denote the relative position between the vehicle and the target, with $\left\| \vect{p} \right\| > 0$. The central bearing ${}^{B}\vect{b}$, expressed in $\{\mathcal{B}\}$ is related to the relative position between the target and the vehicle according to
\begin{equation}
    {}^{B}\vect{b} := \vect{R}^\top\frac{\vect{p}}{\norm{\vect{p}}}.
    \label{eqn:bearing_definition_body_frame}
\end{equation}
The bearing vector can also be expressed in $\{\mathcal{I}\}$, according to $\vect{b}:= \vect{R} \, {}^{B} \vect{b}$. It will also be useful for the further development to consider the angle formed between the central bearing ${}^{B}\vect{b}$ and the tangent bearing ${}^{B}\vect{b}_{t}$ given by
\begin{equation}
    \theta := \arccos({{}^{B}\vect{b}^\top {}^{B}\vect{b}_t}), \text{ with } \theta \in [0, \pi/2).
    \label{eqn:angle_bearings_relation}
\end{equation}
From a combination of geometric arguments and Assumption \ref{assumption:spherical_target}, the distance between the vehicle and the center of the target can also be given by the relation
\begin{equation}
    d = \frac{r}{\sin(\theta)}.
    \label{eqn:range_full_eqn}
\end{equation}

\subsection{Bearing and Angle Dynamics}
The time derivative of the range between the target and the vehicle, expressed in $\{\mathcal{I}\}$ is given by
\begin{equation}
    \dot{d}=\frac{d}{dt}\left( \norm{\vect{p}}\right) = \left(\frac{\vect{p}}{\norm{\vect{p}}}\right)^{\top}\dot{\vect{p}} = \vect{b}^{\top} \vect{v},
    \label{eqn:distance_dynamics_from_norm_derivative}
\end{equation}
where $\vect{v}:=\vect{v}_{\mathrm{T}} - \vect{v}_{\mathrm{B}} \in \mathbb{R}^{3}$ denotes the relative velocity between the vehicle and the target expressed in $\{\mathcal{I}\}$. Taking also the time-derivative of \eqref{eqn:range_full_eqn}, yields the dynamics
\begin{equation}
    \dot{d} = -\frac{r \cos(\theta)}{\sin^{2}(\theta)}\dot{\theta} = -d \cot(\theta)\dot{\theta}.
    \label{eqn:distance_dynamics_from_angle_relation_derivative}
\end{equation}
Since we cannot measure the distance $d$ directly, consider also the introduction of a scaled relative velocity variable $\vect{w} \in \mathbb{R}^3$ given by
\begin{equation}
    \vect{w} := \frac{1}{r}\vect{v},
    \label{eqn:relative_visual_velocity_equation}
\end{equation}
which coupled with \eqref{eqn:distance_dynamics_from_norm_derivative} and \eqref{eqn:distance_dynamics_from_angle_relation_derivative}, gives the following dynamics of the angle formed between the two bearing vectors
\begin{equation}
    \dot{\theta} = -\frac{\sin^{2}(\theta)}{\cos(\theta)}\vect{b}^{\top} \vect{w}.
    \label{eqn:angle_dynamics}
\end{equation}
To derive the dynamics of the central bearing $\vect{b}$ expressed in $\{\mathcal{I}\}$, consider the time-derivative of \eqref{eqn:bearing_definition_body_frame} rotated to the inertial frame, given by
\begin{equation}
    \dot{\vect{b}} = \frac{1}{d}\vectg{\Pi}_\vect{b} \vect{v}.
    \label{eqn:derivative_of_the_bearing}
\end{equation}
By replacing \eqref{eqn:range_full_eqn} and \eqref{eqn:relative_visual_velocity_equation} in \eqref{eqn:derivative_of_the_bearing}, the dynamics of the bearing can be expressed in terms of the angle and the relative velocity, according to
\begin{equation}
    \dot{\vect{b}} =\sin(\theta) \vectg{\Pi}_\vect{b} \vect{w}.
    \label{eqn:bearing_dynamics_angle}
\end{equation}
Finally, the dynamics of the visual relative velocity are given by
\begin{equation}
    \dot{\vect{w}}:=\frac{1}{r}(\vect{a}_{\mathrm{T}} - \vect{u}),
    \label{eqn:relative_velocity_dynamics}
\end{equation}
where $\vect{u} \in \mathbb{R}^3$ is a virtual acceleration input to be designed.

The relative velocity can be decomposed in components parallel and orthogonal to the bearing vector, according to 
\begin{equation}
	\vect{w} = \vectg{\Pi}_\vect{b}\vect{w} + \vect{b}\vect{b}^{\top}\vect{w}.
\end{equation}
Given that measurements of $\dot{\vect{b}}$ and $\dot{\theta}$ can be obtained via discrete differentiation, it follows directly from \eqref{eqn:angle_dynamics} and \eqref{eqn:bearing_dynamics_angle} that a noisy measurement of $\vect{w}$ can be obtained from
\begin{equation}
	\vect{w}^{\text{meas}} := \frac{1}{\sin(\theta)} \dot{\vect{b}} -\vect{b}\frac{\cos(\theta)}{\sin^2(\theta)}\dot{\theta}.
	\label{eqn:measured_relative_velocity}
\end{equation}
The noise in $\vect{w}^{\text{meas}}$ is assumed to be high frequency and have zero bias.

From \eqref{eqn:angle_dynamics} and \eqref{eqn:bearing_dynamics_angle} it becomes clear that the time-derivative of the bearing $\vect{b}$ only depends on the orthogonal component of the velocity $\vectg{\Pi}_\vect{b}\vect{w}$, which corresponds to the tangent space to the sphere at $\vect{b}$. The angle, inversely proportional to the unmeasured range, depends on the parallel component of the relative velocity to the bearing vector $\vect{b}\vect{b}^{\top}\vect{w}$. This yields a new parametrization of the system in polar coordinates, which is suitable for control design, as the target size in the image, dictated by $\theta$, can be velocity controlled independently from the relative direction between the vehicle and the target. Unlike in a Cartesian coordinate space, in this new parametrization, the vehicle is not drawn towards the target when moving along the gradient with the desired relative direction opposite to the current one, according to Fig. \ref{fig:cost}.
\begin{figure}[H]
	\centering
    \includegraphics[width=0.48\textwidth]{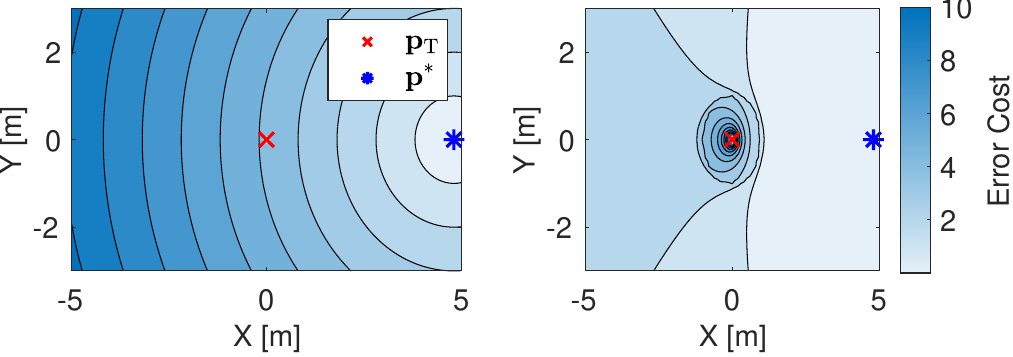}
    \vspace{-0.5cm}
	\caption{Example level sets of the squared norm of the tracking error, parameterized by Cartesian coordinates (left), versus polar coordinates (right), for observing a static target with radius $r=1\,\SI{}{m}$ located at the origin, from a relative reference position $\vect{p}^{\ast}=[-4.8, 0,0]^{\top}$.}
	\label{fig:cost}
\end{figure}
\vspace{-0.5cm}
\subsection{System Dynamics}
Before proceeding with a formal problem definition, consider an additional change of variables given by
\begin{equation}
	x := \sin(\theta),
 \label{eqn:change_of_variables_x_sin_theta}
\end{equation}
such that its time-derivative yields
\begin{equation}
    \dot{x} = \cos(\theta)\dot{\theta}.
    \label{eqn:dynamics_of_scale}
\end{equation}

By replacing \eqref{eqn:angle_dynamics} and \eqref{eqn:change_of_variables_x_sin_theta} in \eqref{eqn:dynamics_of_scale}, the final dynamics of the system can be given by
\begin{equation}
    \begin{cases}
        \dot{\vect{b}} &= x \vectg{\Pi}_\vect{b}\vect{w} \\
        \dot{x} &= -x^2 \vect{b}^{\top}\vect{w} \\
        \dot{\vect{w}} &= \frac{1}{r} (\vect{a}_{\mathrm{T}} - \vect{u}).
    \end{cases}
    \label{eqn:final_dynamics_model}
\end{equation}

This additional change of coordinates not only simplifies the dynamics of the system, but also preserves the general shape of the coordinate space introduced in Fig. \ref{fig:cost} (right). It also bounds the relative distance between the vehicle and the target such that $x \in [0,1)$, where $x=0$ corresponds to the scenario where the vehicle is infinitely far away from the target and $x=1$ to the vehicle touching the surface of the spherical target.

\begin{problem}
Consider the target-tracker system with dynamics described by \eqref{eqn:final_dynamics_model}. Let $\vect{b}^{\ast} \in \mathbb{S}^{2}$ be a constant desired bearing reference and $\theta^{\ast} \in [0, \pi/2)$ the desired relative size of the target in the image, which has a direct mapping to $x^{\ast} \in [0, 1)$. Design a control law that generates a bounded virtual acceleration input $\vect{u} \in \mathbb {R}^{3}$ such that the system tracks the reference signals.
\label{problem:tracker_target_tracking}
\end{problem}
\section{CONTROL DESIGN METHODOLOGY}
\label{sec:Control_design}
In this section, a backstepping nonlinear controller design is derived for performing visual target tracking. Consider a bearing error and a relative distance error tracking terms, given by
%\subsection{Control Strategy}
\begin{equation}
    \begin{cases}
        \vectg{\delta}_1 &:= \vect{b} - \vect{b}^{\ast} \\
        \delta_2 &:= x - x^{\ast}.
    \end{cases}
\end{equation}
The goal of Problem \ref{problem:tracker_target_tracking} is to drive the errors $\vectg{\delta}_1$ and $\delta_2$ asymptotically to zero. Notice, that due to the choice of coordinates, both error variables are naturally bounded, such that $\norm{\vectg{\delta}_1} \in [0, 2)$ and $\delta_2 \in [-1,1]$. 

Taking the time-derivative of $\delta_1$ yields the dynamics of the bearing error given by
\begin{equation}
    \begin{split}
        \dot{\vectg{\delta}}_1 = \dot{\vect{b}} -  \dot{\vect{b}}^{\ast}.
    \end{split}
\end{equation}
Since the bearing reference is assumed to be constant, $
\dot{\vect{b}}^{\ast}=0$ and from the system dynamics \eqref{eqn:final_dynamics_model}, the final bearing error dynamics are given by
\begin{equation}
    \dot{\vectg{\delta}}_1 = x\vectg{\Pi}_\vect{b} \vect{w}.
    \label{eqn:delta_1_dynamics}
\end{equation}
Analogously, the time-derivative of the relative distance error is given by
\begin{equation}
    \dot{\delta}_2 = -x^2 \vect{b}^{\top}\vect{w}.
    \label{eqn:delta_2_dynamics}
\end{equation}

Before proceeding with the control design, we introduce a desired relative velocity $\vect{w}_d \in \mathbb{R}^3$, given by
\begin{equation}
    \vect{w}_d := \frac{k_1}{x} \vectg{\Pi}_\vect{b} \vect{b}^{\ast} + \frac{k_2}{x^2}\delta_2\vect{b},
    \label{eqn:proposed_virtual_desired_velocity}
\end{equation}
where $k_1 , k_2 > 0$ can be regarded as proportional gains along the components orthogonal and parallel to the bearing, respectively.

Consider the Lyapunov candidate function, given by
\begin{equation}
    V_1 := \frac{1}{2}\norm{\vectg{{\delta}}_1}^{2} + \frac{1}{2} \mathbf{\delta}_2^2.
\end{equation}
The time-derivative of $V_1$ is:
\begin{equation}
    \begin{split}
        \dot{V}_1 &= x\vectg{\delta}_{1}^{\top}\vectg{\Pi}_\vect{b}\vect{w} -x^2 \delta_2 \vect{b}^{\top}\vect{w}\\
        &= \vect{w}^{\top}(x\vectg{\Pi}_\vect{b}\vectg{\delta}_{1} - x^{2}\delta_2\vect{b})\\
        &= (\vectg{\delta}_3 + \vect{w}_d)^\top(x\vectg{\Pi}_\vect{b} \vectg{\delta}_{1} - x^2 \delta_2 \vect{b}),
    \end{split}
    \label{eqn:V1_dot}
\end{equation}
where $\vectg{\delta_3}:= \vect{w} - \vect{w}_d \in \mathbb{R}^3$ is a relative velocity tracking error, with dynamics:
\begin{equation}
    \dot{\vectg{\delta}}_3 = \frac{1}{r}(\vect{a}_\mathrm{T} - \vect{u}) - \dot{\vect{w}}_d.
    \label{eqn:delta_3_dynamics}
\end{equation}
To continue the design of the control law, consider a second Lyapunov function
\begin{equation}
    V_2 := V_1 + \frac{1}{2} \norm{\vectg{\delta}_3}^2,
\end{equation}
with time-derivative given by
\begin{equation}
    \begin{split}
         \dot{V}_2 =& \dot{V}_1 + \vectg{\delta}_{3}^{\top} \left(\frac{1}{r} (\vect{a}_{\mathrm{T}} - \vect{u}) - \dot{\vect{w}}_d\right) \\
         %&= \vect{w}^{d\top}(x\vectg{\Pi}_\vect{b} \vect{b}^{\ast} + x^2 \delta_2 \vect{b}) \\ &+ \vectg{\delta}_3^\top (-x\vectg{\Pi}_\vect{b} \vect{b}^{\ast} - x^2 \delta_2 \vect{b} \dot{\vect{w}}_d - \frac{1}{r} (\vect{a}_{\mathrm{T}} - \vect{u})) \\
                   =& \underbrace{-k_1\vectg{\delta}_{1}^{\top}\vectg{\Pi}_\vect{b}\vectg{\delta}_{1} - k_2 \delta_2^2}_{-W(\vectg{\delta}_{1}, \delta_2)} \\
                   &+ \vectg{\delta}_3^{\top}(x\vectg{\Pi}_\vectg{b} \vectg{\delta}_{1} - x^2 \delta_2 \vect{b} - \dot{\vect{w}}_d + \vectg{\rho} - \frac{1}{r}\vect{u}),
    \end{split}
    \label{eqn:V2_dot}
\end{equation}
where $\vectg{\rho}:= \vect{a}_{\mathrm{T}}/r \in \mathbb{R}^3$ is the unknown constant scaled target acceleration. Consider a new virtual control
\begin{equation}
    \vect{u}_0 := \frac{1}{\hat{r}} \vect{u},
    \label{eqn:u0_virtual_input}
\end{equation}
where $\hat{r}$ is an estimation of the unknown constant radius of the target. Consider also two new adaptation error terms
\begin{equation}
    \begin{cases}
        \Tilde{\vectg{\rho}} &:= \vectg{\rho} - \hat{\vectg{\rho}} \\
        \Tilde{r} &:= r - \hat{r}.
    \end{cases}
    \label{eqn:adaptive_estimation_error}
\end{equation}
Replacing \eqref{eqn:u0_virtual_input} and \eqref{eqn:adaptive_estimation_error} in \eqref{eqn:V2_dot} yields
\begin{equation}
    \begin{split}
        \dot{V}_2 =& -W(\vectg{\delta}_1, \delta_2) + \vectg{\delta}_3^{\top}(x\vectg{\Pi}_\vect{b} \vectg{\delta}_1 - x^2 \delta_2 \vect{b} - \dot{\vect{w}}_d \\
        &+ \hat{\vectg{\rho}} + \Tilde{\vectg{\rho}} - \frac{\hat{r}}{r}\vect{u}_0).
    \end{split}
\end{equation}

Let us consider a final Lyapunov function
\begin{equation}
    V_3 := V_2 + \frac{1}{2k_r r} \Tilde{r}^2 + \frac{1}{2} \Tilde{\vectg{\rho}}^{\top}\vect{K}_\vectg{\rho}^{-1}\Tilde{{\vectg{\rho}}},
\end{equation}
with $k_r \in \mathbb{R}^{+}$ and $\vect{K}_\vectg{\rho} \succeq \vect{0}$ corresponding to adaptation gains. Taking the time-derivative of $V_3$, we can verify
\begin{equation}
    \begin{split}
        \dot{V}_3 =& -W(\vectg{\delta}_1, \delta_2) + \Tilde{\vectg{\rho}}^{\top}(\vectg{\delta}_3 -\vect{K}_\vectg{\rho}^{-1}\dot{\hat{\vectg{\rho}}}) + \frac{\Tilde{r}}{r}\left( \vectg{\delta}_3^{\top} \vect{u}_0 - \frac{1}{k_r}\dot{\hat{r}}\right)\\
        &+ \vectg{\delta}_3^{\top}\left(x\vectg{\Pi}_\vect{b} \vectg{\delta}_1 - x^2 \delta_2 \vect{b} - \dot{\vect{w}}_d + \hat{\vectg{\rho}} - \vect{u}_0 \right).
    \end{split}
\end{equation}

\begin{theorem}
    Consider the target-tracker system described by \eqref{eqn:final_dynamics_model} and the error dynamics given by (\ref{eqn:delta_1_dynamics},~\ref{eqn:delta_2_dynamics}) and \eqref{eqn:delta_3_dynamics}. Consider also the relation between the vehicle acceleration and the virtual input \eqref{eqn:u0_virtual_input}. For the control law
    \begin{equation}
        \vect{u}_0 := \hat{\vectg{\rho}} - \dot{\vect{w}}_d - x \vectg{\Pi}_\vect{b} \vect{b}^{\ast} - x^2 \delta_2 \vect{b} + \vect{K}_3 \vectg{\delta}_3,
    \end{equation}
where $\vect{K}_3 \succeq \vect{0}$ a gain matrix, with adaptive observers
\begin{equation}
    \begin{cases}
        \dot{\hat{r}} &:= k_r \vectg{\delta}_{3}^{\top}\vect{u}_0 \\
        \dot{\hat{\vectg{\rho}}} &:= \vect{K}_\vectg{\rho} \vectg{\delta}_3
    \end{cases},
\end{equation}
there exists a compact set of initial conditions where target tracking is achieved, guaranteeing that the errors $\vectg{\delta}_1$, $\delta_2$ and $\vectg{\delta}_3$ converge to zero.
\end{theorem}

\begin{proof}
Starting with the positive definite Lyapunov function $V_3$, and using the proposed $\vect{u}_0$, $\dot{\hat{r}}$ and $\dot{\hat{\vectg{\rho}}}$, its closed-loop time-derivative is given by
\begin{equation}
    \dot{V}_3 = -k_1\vectg{\delta}_{1}^{\top} \vectg{\Pi}_\vect{b} \vectg{\delta}_{1} - k_2 \delta_2^2 - \vectg{\delta}_{3}^{\top} \vect{K}_3 \vectg{\delta}_{3} \leq 0,
\end{equation}
which is negative semi-definite on the error states. Since the system error dynamics are non-autonomous, we resort to Barbalat's Lemma to prove the convergence of $\dot{V}_3$ to zero. Since the $V_3$ is positive definite with respect to the states $\vectg{\delta}_1, \delta_2$ and $\vectg{\delta}_3$, and the adaptation error terms $\vectg{\Tilde{\rho}}$ and $\Tilde{r}$, and $\dot{V}_3$ is negative semi-definite, we can conclude that the state errors are bounded. From the boundedness of the state errors, we can also conclude that the estimates $\hat{\vectg{\rho}}$ and $\hat{r}$, and the control input $\vect{u}$ are also bounded. Computing $\ddot{V}_3$, shows that the second time-derivative of the Lyapunov function only depends on bounded terms, hence it is also bounded, from which follows that the time derivative $\dot{V}_3$ is uniformly continuous. By direct application of Barbalat's Lemma, we can conclude that $\dot{V}_3$ converges to zero and, the errors $\vectg{\Pi}_\vect{b}\vectg{\delta}_1, \delta_2$ and $\vectg{\delta}_3$ converge asymptotically to the origin. Given that all errors are bounded, we can conclude that there exists a compact set of initial conditions, such that $\vectg{\delta}_1$ is guaranteed to converge to zero \cite{BHAT200063}. However, convergence of the the adaptive estimation terms to their real values cannot be concluded.
\end{proof}
\begin{remark}
    In practice, the inclusion of $\dot{\vect{w}}_d$ in the proposed control law is optional as it may introduce unnecessary noise in the system, as a consequence of assuming that $\dot {\vect{b}}$ and $\dot{\theta}$ are measured via discrete differentiation. Since the reference bearing and relative distance are constant, $\vect{w}_d$ will be zero at the origin, and the remaining terms on time-derivative of the Lyapunov function that depend on $\dot{\vect{w}}_d$ will also vanish to zero at the origin.
    \label{remark:w_d_dot}
\end{remark}
\section{Target Visibility Constraints}
\label{sec:target_visibility_constraints}
In the previous section, we derived a control law that generates a desired virtual acceleration vector $\vect{u}$ to be tracked. In ideal conditions, the thrust vector of the multirotor would align with this input vector. However, in most application scenarios a camera has a limited field-of-view that must be accounted for, leading to the following assumption.
\begin{assumption}
	The camera sensor is fixed, forward-facing and aligned with the vehicle's body frame $\{\mathcal{B}\}$. Moreover, the reference bearing vector $\vect{b}^{\ast}$ is such that the target is in the field-of-view of the vehicle.
	\label{assumption:monocular_fixed_camera}
\end{assumption}

To take into account the target visibility limitations, the desired attitude of the vehicle can be decomposed into three orthogonal components, according to
\begin{equation}
	\vect{R}_{des} := \begin{bmatrix}\vect{x}_{\mathrm{B}}^{des} & \vect{y}_{\mathrm{B}}^{des} & \vect{z}_{\mathrm{B}}^{des}\end{bmatrix},
\end{equation}
where the desired z-axis of the vehicle is dictated by the control law derived in the previous section, according to
\begin{equation}
	\vect{z}_{\mathrm{B}}^{\ast}:= - \frac{\vect{u} - g \vect{e}_3}{\norm{\vect{u} - g \vect{e}_3}}.
\end{equation}
Given that the yaw-angle is an extra degree of freedom, we can always choose an arbitrary $\vect{x}_{\mathrm{B}}$ such that the target is always within an horizontal field-of-view of the vehicle. In this case, we select $\vect{x}_{\mathrm{B}}$ such that it aligns vertically with the bearing vector. To achieve this, start by considering that the desired y-axis of the vehicle is given by
\begin{equation}
	\vect{y}_{\mathrm{B}}^{des}:=\frac{\vect{z}_{\mathrm{B}}^{\ast} \times \vect{b}}{\norm{\vect{z}_{\mathrm{B}}^{\ast} \times \vect{b}}}.
\end{equation}

To encode the vertical visibility constraint, we consider two virtual cones along the $\vect{z}_{\mathrm{B}}$ axis that encode visibility dead-zones, according to Fig. \ref{fig:virtual_system_architecture}. In order to guarantee that the target is always visible, the following constraint must be satisfied
\begin{equation}
	-\cos(\varphi) \leq \vect{b}^{\top}\vect{z}_{\mathrm{B}} \leq \cos(\varphi),
    \label{eqn:vertical_visibility_constraint}
\end{equation}
where $\varphi \in [0, \pi/2)$ is the vertex angle that defines the size of the dead-zone. To ensure that \eqref{eqn:vertical_visibility_constraint} is satisfied, an extra rotation about $\vect{y}_{\mathrm{B}}^{des}$ might be necessary. Consider the final $\vect{z}_{\mathrm{B}}^{des}$ reference to be given by 
\begin{equation}
	\vect{z}_{\mathrm{B}}^{des} := \vect{R}(\psi, \vect{y}_{\mathrm{B}}^{des})\vect{z}_{\mathrm{B}}^{\ast},
\end{equation}
where $\vect{R}(\psi, \vect{y}_{\mathrm{B}}^{des}) \in SO(3)$ is a rotation matrix, given by the Rodrigues' formula, according to
\begin{equation}
\resizebox{.46\textwidth}{!}{$
\begin{split}
	\vect{R}(\psi, \vect{y}_{\mathrm{B}}^{des}) :=& \vect{I} + \sin(\psi)\vectg{S}(\vect{y}_{\mathrm{B}}^{des}) + (1-\cos(\psi))\vectg{S}^{2}(\vect{y}_{\mathrm{B}}^{des}).
\end{split}$}
\end{equation}
The rotation angle $\psi \in [0, \pi]$ is given by
\begin{equation}
\resizebox{.42\textwidth}{!}{$
    \psi := 
    \begin{cases}
       \varphi-\arccos(- \vect{b}^{\top}\vect{z}_{\mathrm{B}}^{\ast}) & \text{,}   \vect{b}^{\top}\vect{z}_{\mathrm{B}}^{\ast} \leq -\cos(\varphi) \\
    \arccos(\vect{b}^{\top}\vect{z}_{\mathrm{B}}^{\ast})-\varphi & \text{,}   \vect{b}^{\top}\vect{z}_{\mathrm{B}}^{\ast} \geq \cos(\varphi) \\
        0 & \text{,} \text{otherwise} \\
    \end{cases}$}.
\end{equation}
\begin{figure}[H]
	\centering
 \includegraphics[width=0.43\textwidth]{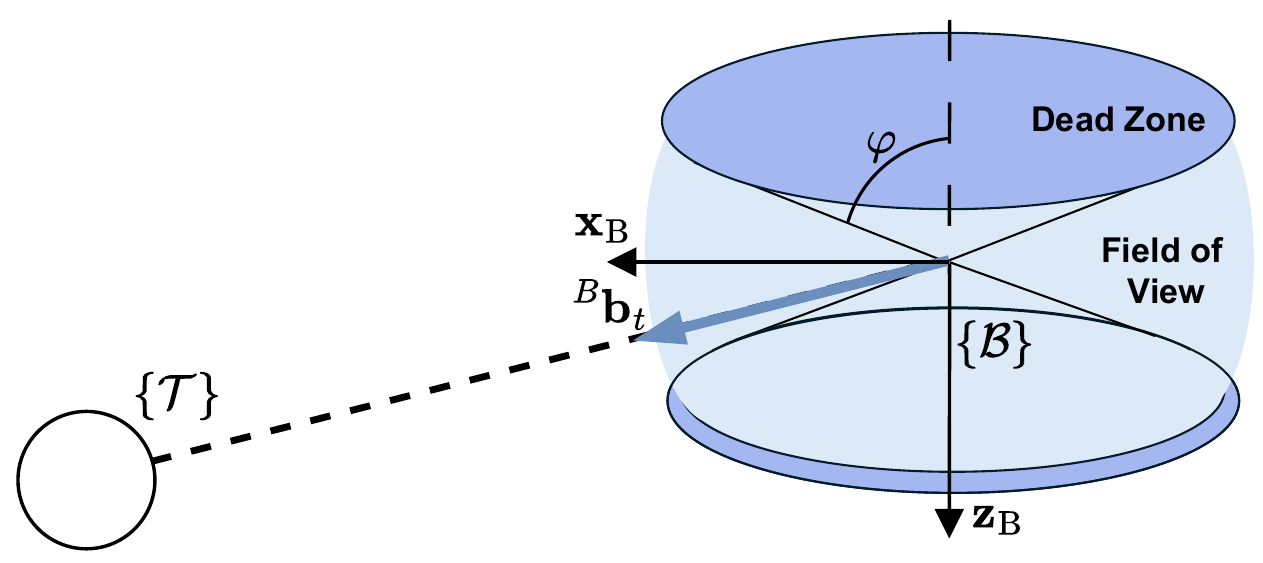}
    \vspace{-0.1cm}
	\caption{Illustration of the visibility model, where the volume in light blue encodes the field of view of the camera sensor.}
	\label{fig:virtual_system_architecture}
\end{figure}
\vspace{-0.3cm}
Finally, we can define the desired $\vect{x}_{\mathrm{B}}^{des}$ component of the attitude to be tracked, according to
\begin{equation}
    \vect{x}_{\mathrm{B}}^{des} := \vect{y}_{\mathrm{B}}^{des} \times \vect{z}_{\mathrm{B}}^{des}.
\end{equation}
The total force desired to be applied on the vehicle can be obtained directly from the desired acceleration control input, according to:
\begin{equation}
    \vect{F}^{\text{des}}:= m(\vect{u} - g \vect{e}_3),
\end{equation}
from which the total thrust $T$ to be applied to the vehicle is given by
\begin{equation}
    T := -\vect{z}_{\mathrm{B}}^{\top}\vect{F}^{\text{des}}.
\end{equation}

To track the desired attitude, one can define the orientation error $\vect{e}_{\mathrm{R}} \in \mathbb{R}^{3}$ given by
\begin{equation}
    \vect{e}_{\mathrm{R}}:=\frac{1}{2}\vectg{S}^{-1}(\vect{R}_{des}^{\top}\vect{R}-\vect{R}^{\top}\vect{R}_{des}).
\end{equation}
The angular-velocity used as input to the vehicle is given by
\begin{equation}
    \vectg{\omega} := -K_{R}\vect{e}_{\mathrm{R}},
\end{equation}
where $K_{R} \succeq 0$ is an attitude gain.

\section{Numerical Simulation}
\label{sec:simulation_results}
To evaluate the performance of our theoretical results, this section presents numerical simulations representative of the closed-loop target-tracker system using MATLAB\textsuperscript{\tiny\textregistered} and SIMULINK. In this example, the vehicle is set to track an accelerating target, with radius $r = 0.25 \,\SI{}{m}$ from a fixed reference direction $\vect{b}^{\ast}=[-1.0, 0.001, 0.0]^{\top}$ with relative angle $x^{\ast}\approx \theta^{\ast}=0.125 \, \SI{}{rad}$, according to Fig. \ref{fig:results_3d}. The initial position of the vehicle and the target are $\vect{p}_{\mathrm{B}}(t_0)=[0, 0, -1.8]^{\top} \SI{}{m}$ and $\vect{p}_{\mathrm{T}}(t_0)=[3, 0.1, -1.0]^{\top} \SI{}{m}$, respectively. Both vehicle and target start with zero velocity, $\vect{v}_\mathrm{B}(t_0)=\vect{v}_\mathrm{T}(t_0)= [0, 0, 0]^{\top} \SI{}{m/s}$, but the target is accelerating, according to $\vect{a}_\mathrm{T}(t_0) = [-0.01, 0.01, 0]^{\top} \SI{}{m/s^2}$.

The bearing vector measurements were corrupted by rotation noise with an angle standard deviation of $1^{\circ}$. Analogously, the angle measurements were also corrupted by noise with standard deviation $10^{-4} \, {}^{\circ}$. The measured relative velocity is computed from the differentiation of these measurements, according to \eqref{eqn:measured_relative_velocity}. The simulated vehicle has a narrow vertical field-of-view with $\varphi=75^{\circ}$, a total mass of $m=1.0 \, \SI{}{Kg}$ and a maximum thrust $T_{max}=34 \, \SI{}{N}$. The control gains considered are $k_1 = 0.4$, $k_2 = 1.2$, $\vect{K}_3 = 0.7\vect{I}_{3}$ and $\vect{K}_R = 5.0\vect{I}_{3}$, with adaptive observer gains $\vect{K}_{\vectg{\rho}}=10^{-4}\vect{I}_{3}$ and $k_r=0.1$. The adaptive observers were initialized with $\hat{r}(t_0)=1.0$ and $\hat{\vectg{\rho}}(t_0)=[0, 0, 0]^{\top}$.
\begin{figure}[H]
	\centering
 \includegraphics[width=0.45\textwidth]{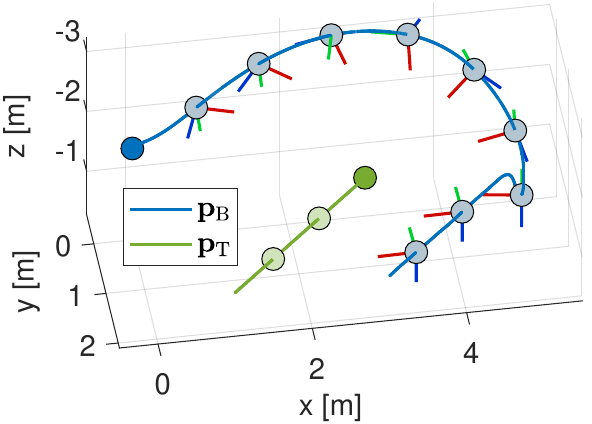}
	\caption{Simulation of the target-tracker system with a constant bearing-angle reference. The vehicle and tracker positions are depicted in blue and green, respectively. Each circle represents a time step, with the attitude of the vehicle body axis represented using red, green and blue vectors, corresponding to $\vect{x}_{\mathrm{B}}$, $\vect{y}_{\mathrm{B}}$ and $\vect{z}_{\mathrm{B}}$, respectively.}
	\label{fig:results_3d}
\end{figure}
\vspace{-0.2cm}

The performance of the controller is shown in Fig. \ref{fig:errors_along_time}, where it is observed that bearing tracking error $\vectg{\delta}_1$, the relative distance error $\delta_2$ and the relative velocity error $\vectg{\delta}_3$ converge to the origin in finite time. It can be observed that the desired relative velocity also approaches zero as the vehicle approaches the desired reference. Fig. \ref{fig:errors_along_time} also shows that although the tracking errors converge to zero, the adaptive estimators $\hat{\vectg{\rho}}$ and $\hat{r}$ converge to constant values that do not correspond to real physical values, in accordance with the stability analysis.
\begin{figure}
	\centering
 \includegraphics[width=0.50\textwidth]{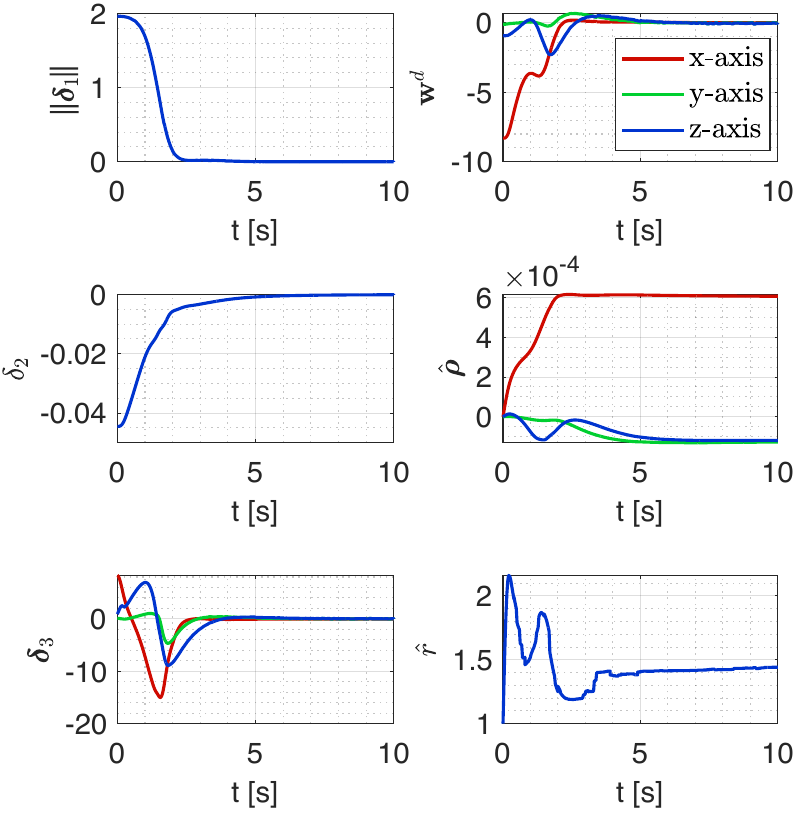}
	\caption{Evolution of the bearing tracking error norm $\norm{\vectg{\delta}_1}$, the relative distance error $\delta_2$, relative velocity error $\vectg{\delta}_3$, desired relative velocity $\vect{w}^{d}$ and the adaptive estimators $\hat{\vectg{\rho}}$ and $\hat{\vect{r}}$.}
	\label{fig:errors_along_time}
\end{figure}
\section{CONCLUSION}
\label{sec:conclusion}
This work introduced a novel approach for tracking a moving target that can be approximated by a sphere with an unknown radius. By decoupling the system state into a bearing-angle pair, the nonlinear dynamics of the system expressed in polar coordinates naturally emerged. This transformation led to a new coordinate system practical for real-world scenarios: it allowed  us to use a reference angle to dictate the desired size of the target within the image and a bearing for its relative direction. The geometric properties of this transformed system were then exploited in the design of an IBVS inspired nonlinear control law, enabling  tracking the target from a fixed relative position. Finally, field-of-view limitations were incorporated as input constraints, which is essential for real-world applications. Simulation results showcased the performance of the proposed algorithm, with real test flights planned for future work.

\section*{Acknowledgment}
The authors gratefully acknowledge J. Pinto, P. Santos and Prof. T. Hamel for their suggestions that helped improve the quality of this work.
\vspace{-0.2cm}
\bibliographystyle{IEEEtran}
\bibliography{IEEEabrv, bibliography}

\end{document}